\documentclass[11pt]{article}
\usepackage{inputenc}
\usepackage[T1]{fontenc}
\usepackage[letterpaper,margin=0.90in]{geometry}
\usepackage{amsmath, amssymb, amsthm, thmtools, amsfonts}
\usepackage{bbm}

\usepackage{ifthen}

\usepackage{tikz}
\usetikzlibrary{positioning,decorations.pathreplacing}

\usepackage{cite}
\usepackage{appendix}
\usepackage{graphicx}
\usepackage{color}
\usepackage{algorithm}
\usepackage[noend]{algpseudocode}
\usepackage{epstopdf}

\usepackage{subcaption}
\usepackage{framed}
\usepackage[framemethod=tikz]{mdframed}
\usepackage[bottom]{footmisc}
\usepackage[shortlabels]{enumitem}
\setitemize{noitemsep,topsep=3pt,parsep=3pt,partopsep=3pt}
\usepackage[font=small]{caption}
\usepackage{xspace}

\newtheorem{theorem}{Theorem}[section]
\newtheorem{lemma}[theorem]{Lemma}
\newtheorem{meta-theorem}[theorem]{Meta-Theorem}

\newtheorem{remark}[theorem]{Remark}
\newtheorem{corollary}[theorem]{Corollary}

\newtheorem{definition}[theorem]{Definition}

\definecolor{darkgreen}{rgb}{0,0.5,0}
\usepackage{hyperref}
\hypersetup{
    unicode=false,          
    colorlinks=true,        
    linkcolor=red,          
    citecolor=darkgreen,        
    filecolor=magenta,      
    urlcolor=cyan           
}
\usepackage[capitalize, nameinlink]{cleveref}
\crefname{theorem}{Theorem}{Theorems}
\Crefname{lemma}{Lemma}{Lemmas}
\Crefname{figure}{Figure}{Figures}
\Crefname{claim}{Claim}{Claims}
\Crefname{observation}{Observation}{Observations}
\Crefname{definition}{Defintion}{Definitions}
\algnewcommand\algorithmicswitch{\textbf{switch}}
\algnewcommand\algorithmiccase{\textbf{case}}

\algdef{SE}[SWITCH]{Switch}{EndSwitch}[1]{\algorithmicswitch\ #1\ \algorithmicdo}{\algorithmicend\ \algorithmicswitch}%
\algdef{SE}[CASE]{Case}{EndCase}[1]{\algorithmiccase\ #1}{\algorithmicend\ \algorithmiccase}%
\algtext*{EndSwitch}%
\algtext*{EndCase}%

\newcommand{\eps}{\varepsilon}

\newcommand{\local}{$\mathsf{LOCAL}$\xspace}
\newcommand{\pram}{$\mathsf{PRAM}$\xspace}
\newcommand{\mpc}{$\mathsf{MPC}$\xspace}

\newcommand{\poly}{\operatorname{\text{{\rm poly}}}}

\newcommand{\arb}{\lambda}
\newcommand{\param}{d}

\renewcommand{\paragraph}[1]{\vspace{0.15cm}\noindent {\bf #1}:}

\newcommand{\FullOrShort}{full}

\ifthenelse{\equal{\FullOrShort}{full}}{
	
  \newcommand{\fullOnly}[1]{#1}
  \newcommand{\shortOnly}[1]{}
  
  }{

    \newcommand{\fullOnly}[1]{}
    \newcommand{\IncludePictures}[1]{}
   
  }

\usepackage[textsize=tiny]{todonotes}

\begin{document}

\title{Matching and MIS for Uniformly Sparse Graphs \\ in the Low-Memory \mpc Model}
\date{}
\author{
	Sebastian Brandt \\
	ETH Zurich \\
	brandts@ethz.ch
	\and
	 Manuela Fischer\\
   ETH Zurich \\
   manuela.fischer@inf.ethz.ch
	\and
	 Jara Uitto\thanks{supported by ERC Grant No. 336495 (ACDC)}\\
   ETH Zurich \& \\ University of Freiburg \\
	 jara.uitto@inf.ethz.ch
 }

\maketitle

\setcounter{page}{0}
\thispagestyle{empty}

\begin{abstract}
The \emph{Massively Parallel Computation (\mpc)} model serves as a common abstraction of many modern large-scale parallel computation frameworks and has recently gained a lot of importance, especially in the context of classic graph problems. Unsatisfactorily, all current $\poly (\log \log n)$-round \mpc algorithms seem to get fundamentally stuck at the linear-memory barrier: their efficiency crucially relies on each machine having space at least linear in the number $n$ of nodes. As this might not only be prohibitively large, but also allows for easy if not trivial solutions for sparse graphs, we are interested in the \emph{low-memory \mpc} model, where the space per machine is restricted to be strongly sublinear, that is, $n^{\delta}$ for any $0<\delta<1$. 

We study \emph{maximal matching} and \emph{maximal independent set} in this low-memory \mpc model. Our key ingredient is a degree reduction technique that reduces these problems in graphs with arboricity $\lambda$ to the corresponding problems in graphs with maximum degree $\poly(\arb)$ in $O(\log^2 \log n)$ rounds. This gives rise to $O\left(\log^2\log n + T(\poly \lambda)\right)$-round algorithms, where $T(\Delta)$ is the $\Delta$-dependency in the round complexity of maximal matching and maximal independent set in graphs with maximum degree $\Delta$. A direct simulation of the \local algorithms by Barenboim et al.\ [FOCS'12] and Ghaffari [SODA'16] leads to $T(\Delta)=O(\log \Delta)$, and a concurrent work by Ghaffari and Uitto shows that $T(\Delta)=O(\sqrt{\log \Delta})$. 

For graphs with arboricity $\arb=\poly(\log n)$, this almost exponentially improves over Luby's $O(\log n)$-round \pram algorithm [STOC'85, JALG'86], and constitutes the first $\poly (\log \log n)$-round maximal matching algorithm in the low-memory \mpc model, thus breaking the linear-memory barrier. Previously, the only known subpolylogarithmic algorithm, due to Lattanzi et al.\ [SPAA'11], required strongly superlinear, that is, $n^{1+\Omega(1)}$, memory per machine. 

Moreover, our maximal matching algorithm can be employed to obtain a $(1+\eps)$-approximate \emph{maximum cardinality matching}, a $(2+\eps)$-approximate \emph{maximum weighted matching}, as well as a $2$-approximate \emph{minimum vertex cover} in essentially the same number of rounds. 
\end{abstract}

\newpage

\section{Introduction}	
Due to the prevalence of huge networks, scalable algorithms for fundamental graph problems recently have gained a lot of importance in the area of parallel computing. 
The \emph{Massively Parallel Computation (\mpc)} model \cite{karloff2010model,goodrich2011sorting, beame2014skew, Andoni2014, beame2017communication, czumaj2017round} constitutes a common abstraction of several popular large-scale computation frameworks---such as MapReduce \cite{dean2008mapreduce}, Dryad \cite{isard2007dryad}, Spark \cite{zaharia2010spark}, and Hadoop \cite{white2012hadoop}---and thus serves as the basis for the systematic study of massively parallel algorithms.

While classic parallel (e.g., \pram) or distributed (e.g., \local) algorithms can usually be implemented straightforwardly in \mpc in the same number of rounds \cite{karloff2010model,goodrich2011sorting}, 
the additional power of local computation (compared to \pram) or of global communication (compared to \local) could potentially be exploited to obtain faster \mpc algorithms. 
Czumaj et al. \cite{czumaj2017round} thus ask:
\begin{center}
\begin{minipage}{0.95\linewidth}
\begin{mdframed}[hidealllines=true, backgroundcolor=gray!00]
\vspace{-0.5pt}
\emph{``Are the \mpc parallel round bounds ``inherited'' from the \pram model
tight? \\ In particular, which problems can be solved in significantly smaller number of \mpc rounds 
\\ than
what the [...] \pram model suggest[s]?''}  
\vspace{-0.5pt}
\end{mdframed}
\end{minipage}
\end{center}
Surprisingly, for \emph{maximal matching}, one of the most central graph problems in parallel and distributed computing that goes back to the very beginning of the area, the answer to this question is not known. In fact, even worse,
our understanding of the maximal matching problem in the \mpc model is rather bleak. 
Indeed, the only subpolylogarithmic\footnote{We are aware of an $O(\sqrt{\log n})$-round low-memory \mpc algorithm by Ghaffari and Uitto in a concurrent work \cite{GU18}. However, only algorithms that are significantly faster than their \pram counterparts---that is, smaller than any polynomial in $\log n$, and ideally at most a polynomial in $\log \log n$---are considered efficient.} algorithm by Lattanzi, Moseley, Suri, and Vassilvitskii \cite{lattanzi2011filtering} requires the memory per machine to be substantially superlinear in the number $n$ of nodes in the graph. This is not only prohibitively large and hence impractical for massive graphs, but also allows an easy or even trivial solution for sparse graphs, which both are often assumed to be the case for many practical graphs\cite{karloff2010model,czumaj2017round}. When the local memory is restricted to be (nearly) linear in $n$, the round complexity of Lattanzi et al.'s algorithm drastically degrades, falling back to the trivial bound attained by the simulation of the $O(\log n)$-round \pram algorithm due to Luby \cite{luby1985simple} and, independently, Alon, Babai, and Itai \cite{alon1986fast}. 

For other classic problems, such as \emph{Maximal Independent Set (MIS)} or  \emph{approximate maximum matching}, we have a slightly better understanding. There, the local memory can be reduced to be $\widetilde{O}(n)$ while still having $\poly (\log \log n)$-round algorithms \cite{assadi2017coresets,czumaj2017round,MPCMIS}. Yet, all these algorithms fail to go (substantially) below linear space without an (almost) exponential blow-up in the running time. 
It is thus natural to ask
whether there is a fundamental reason why the
known techniques get stuck at the linear-memory barrier
and study the \emph{low-memory} \mpc model with strongly sublinear space to address this question. 

\paragraph{Low-Memory \mpc Model for Graph Problems} We have $M$ machines with local memory of 
$S=O\left(n^{\delta}\right)$ words each, for some $0<\delta\leq 1$\footnote{Note that we do not require $\delta$ to be a constant. For the sake of simplicity of presentation, we decided to omit the $\delta$-dependency, which is a multiplicative factor of $O\left(\frac{1}{\delta}\right)$, in the analysis of the running time of our algorithms.}. A graph with $n$ nodes, $m$ edges, and maximum degree $\Delta$ is distributed arbitrarily across the machines. 
We assume the total memory in the system to be (nearly) linear in the input, i.e., $M\cdot S = \widetilde{\Theta}(m)$. 
The computation proceeds in rounds consisting of \emph{local computation} in all machines in parallel, followed by \emph{global communication} between the machines. We require that the total size of sent and received messages of a machine in every communication phase does not exceed its local memory capacity. The main interest lies in minimizing the number of rounds, aiming for $\poly (\log \log n)$. 

\newpage
In this low-memory \mpc model, 
the best known algorithms usually stem from straight forward simulations of \pram and \local algorithms, thus requiring at least polylogarithmic rounds. In the special case of trees, \cite{brandt2018breaking} managed to beat this bound 
by providing an $O(\log^3 \log n)$-round algorithm for MIS. The authors left it as a main open question whether there is a low-memory \mpc algorithm for general graphs in $\poly (\log \log n)$ rounds. 

We make a step towards answering this question by devising a degree reduction technique that reduces the problems of maximal matching and maximal independent set in a graph with arboricity $\lambda$ to the corresponding problems in graphs with maximum degree $\poly(\lambda)$ in $O\left(\log^2 \log n\right)$ rounds.

\begin{theorem}\label{thmDegRed}There is an $O\left(\log \log_{\Delta} n \cdot \log \log_{\lambda} \Delta\right)$-round low-memory \mpc algorithm that w.h.p.\footnote{As usual, w.h.p.\ stands for \emph{with high probability}, and means with probability at least $1-n^{-c}$ for any constant $c\geq 1$.} reduces maximal matching and maximal independent set in graphs with arboricity $\lambda=o(\poly (n))$\footnote{With $o(\poly(n))$ we mean subpolynomial in $n$, i.e., $o(n^c)$ for any constant $c>0$. Note that for $\lambda=\poly(n)$, our results for MIS and maximal matching follow directly, without any need for degree reduction.} to the respective problems in graphs with maximum degree $O\left(\max\{\lambda^{20}, \log^{20} n\}\right)$\footnote{The purpose of the choice of all the constants in this work is merely to simplify presentation.}. 
\end{theorem}
This improves over the degree reduction algorithm by Barenboim, Elkin, Pettie, and Schneider \cite[Theorem 7.2]{barenboim2012locality}, which runs in $O(\log_{\lambda} \Delta)$ rounds 
 in the \local model and can be straightforwardly implemented in the low-memory \mpc model.

Our degree reduction technique, combined with the state-of-the-art algorithms for maximal matching and maximal independent set, 
gives rise to a number of low-memory \mpc algorithms, as overviewed next. Throughout, we state our running times based on a concurrent work by Ghaffari and Uitto \cite{GU18}, in which they prove that maximal matching and maximal independent set can be solved in $O\left(\sqrt{\log\Delta}+\log \log \log n\right)$ and $O\left(\sqrt{\log\Delta}+\sqrt{\log\log n}\right)$ low-memory \mpc rounds, respectively, improving over $O\left(\log\Delta+\log \log \log n\right)$ and $O\left(\log\Delta+\sqrt{\log\log n}\right)$, respectively, which can be obtained by a sped up simulation \cite{diam} of state-of-the-art \local algorithms \cite{barenboim2012locality,Ghaffari-MIS}. We apply these algorithms by Ghaffari and Uitto as a black box. 
By using earlier results, the term $O(\sqrt{\log \lambda})$ in the running times of our theorem statements would get replaced by $O(\log \lambda)$. 

\begin{theorem}\label{thmMM}
There is an $O\left(\sqrt{\log \arb}+\log\log n \cdot \log \log \Delta\right)$-round low-memory \mpc algorithm that w.h.p.\ computes a maximal matching in a graph with arboricity $\arb$.
\end{theorem}
This improves over the $O\left(\log\arb+\sqrt{\log n}\right)$-round \local algorithm by \cite{barenboim2012locality} and the $O(\sqrt{\log \Delta}+\log \log \log n)$-round algorithm by \cite{GU18}. 
We get the first $\poly (\log \log n)$-round algorithm---and hence an almost exponential improvement over the state of the art, for linear as well as strongly sublinear space per machine---for all graphs with arboricity $\arb=\poly (\log n)$. This family of \emph{uniformly sparse} graphs, also known as \emph{sparse everywhere} graphs, includes but is not restricted to graphs with maximum degree $\poly (\log n)$, minor-closed graphs (e.g., planar graphs and graphs with bounded treewidth), and preferential attachment graphs, and thus arguably contains most graphs of practical relevance \cite{goel2006bounded,OnakFullyDynamicMIC}.
The previously known $\poly(\log \log n)$-round \mpc algorithms either only worked in the special case of $\poly (\log n)$-degree graphs \cite{GU18}, or required the local memory to be strongly superlinear \cite{lattanzi2011filtering}. 

To the best of our knowledge, all $\poly(\log \log n)$-round \mpc matching approximation algorithms for general graphs (even if we allow linear memory) \cite{czumaj2017round,assadi2017coresets,MPCMIS} heavily make use of subsampling,
 inevitably leading to a loss of information. It is thus unlikely that these techniques are applicable for maximal matching, at least not without a factor $\Omega(\log n)$ overhead. In fact, the problem of finding a maximal matching seems to be more difficult than finding a $(1+\eps)$-approximate maximum matching. Indeed, currently, an $O(1)$-approximation can be found almost exponentially faster than a maximal matching, and the approximation ratio can be easily improved from any constant to $1+\eps$ using a reduction of McGregor \cite{mcgregor2005finding}. 

Our result becomes particularly instructive when viewed in this context. It is not only the first maximal matching algorithm that breaks the linear-memory barrier for a large range of graphs, but it also enriches the bleak pool of techniques for maximal matching in the presence of low memory by one particularly simple technique. 

\begin{theorem}\label{thmMIS}
There is an $O\left(\sqrt{\log \arb}+\log\log n \cdot \log \log \Delta\right)$-round low-memory \mpc algorithm that w.h.p.\ computes a maximal independent set in a graph with arboricity $\arb$. 
\end{theorem}
This algorithm improves over the $O\left(\log\arb+\sqrt {\log n}\right)$-round algorithm that is obtained by simulating the \local algorithm of \cite{barenboim2012locality,Ghaffari-MIS} and over the $O\left(\sqrt{\log \Delta} + \sqrt{\log \log n}\right)$-round algorithm in the concurent work by Ghaffari and Uitto \cite{GU18}. 
Moreover, for graphs with arboricity $\lambda=\poly (\log n)$, our algorithm is the first $\poly (\log \log n)$-round low-memory \mpc algorithm. The previously known $\poly(\log \log n)$-round \mpc algorithms for MIS either only worked in the special case of $\poly (\log n)$-degree graphs \cite{GU18} and trees \cite{brandt2018breaking}, or required the local memory to be $\widetilde{\Omega}(n)$ \cite{MPCMIS}.

\vspace{0.5cm}
As a maximal matching automatically provides $2$-approximations for maximum matching and minimum vertex cover, \Cref{thmMM} directly implies the following result.

\begin{corollary}\label{corMandVC}
There is an $O\left(\sqrt{\log \arb}+\log\log n \cdot \log \log \Delta\right)$-round low-memory \mpc algorithm that w.h.p.\ computes a $2$-approximate maximum matching and a $2$-approximate minimum vertex cover in a graph with arboricity $\arb$. 
\end{corollary}
This is the first constant-approximation \mpc algorithm for matching and vertex cover---except for the \local and \pram simulations by \cite{barenboim2012locality} and \cite{luby1985simple,alon1986fast}, respectively, as well as the concurrent work in \cite{GU18}---that work with low memory. All the other algorithms require the space per machine to be either $\widetilde{\Omega}(n)$ \cite{lattanzi2011filtering,czumaj2017round,assadi2017coresets,MPCMIS} or even strongly superlinear \cite{lattanzi2011filtering,AssadiK17}. \Cref{corMandVC} generalizes the range of graphs that admit an efficient constant-approximation for matching and vertex cover in the low-memory \mpc model from graphs with maximum degree $\poly(\log n)$ \cite{GU18} to uniformly sparse graphs with arboricity $\lambda=\poly(\log n)$. 
\vspace{0.2cm}

McGregor's reduction \cite{mcgregor2005finding} allows us to further improve the approximation to $1+\eps$. 
\begin{corollary}\label{cor1+eps}
There is an $O\left(\left(\frac{1}{\eps}\right)^{O(1/\eps)}\cdot \left( \sqrt{\log \arb}+\log\log n \cdot \log \log \Delta\right)\right)$-round low-memory \mpc algorithm that w.h.p.\ computes a $(1+\eps)$-approximate maximum matching, for any $\eps>0$.
\end{corollary}
\vspace{0.2cm}

Due to a reduction by to Lotker, Patt-Shamir, and Ros\'en \cite{lotkerMatching}, our constant-approximate matching algorithm can be employed to find a $(2+\eps)$-approximate maximum weighted matching. 
\begin{corollary}\label{corW2+eps}
There is an $O\left(\frac{1}{\eps}\cdot \left(\sqrt{\log \arb}+\log\log n \cdot \log \log \Delta\right)\right)$-round low-memory \mpc algorithm that w.h.p.\ computes a $(2+\eps)$-approximate maximum weighted matching, for any $\eps>0$.
\end{corollary} 
\vspace{0.2cm}

\subsection*{Concurrent Work}
In this section, we briefly discuss an independent and concurrent work by Behnezhad, Derakhshan, Hajiaghayi, and Karp \cite{concurr}. The authors there arrive at the same results, with the same round complexities and the same memory requirements. As in our work, their key ingredient is a degree reduction technique that reduces the maximum degree of a graph from $\Delta$ to $\poly (\lambda, \log n)$ in $O(\log \log \Delta \cdot \log \log n)$ rounds. While our degree reduction algorithm is based on (a variant of) the $H$-partition that partitions the vertices according to their degrees, the authors in \cite{concurr} show how to implement (and speed up) the \local degree reduction algorithm by Barenboim, Elkin, Pettie, and Schneider \cite[Theorem 7.2]{barenboim2012locality} in the low-memory \mpc model. Note that both algorithms do not require knowledge of the arboricity $\lambda$, as further discussed in \Cref{remarkKnowledgeArb}.
\newpage 
\section{Algorithm Outline and Roadmap}

In the low-memory setting, one is inevitably confronted with 
the challenge of locality: 
as the space of a machine is strongly sublinear, it will never be able to see a significant fraction of the nodes, regardless of how sparse the graph is. 
Further building on the ideas by \cite{brandt2018breaking}, we cope with this imposed locality by adopting local techniques---mainly inspired by the \local model \cite{linial1992locality}---and enhancing them with the additional power of global communication, in order to achieve an improvement in the round complexity compared to the \local algorithms while still being able to guarantee applicability 
in the presence of strongly sublinear memory.

The main observation behind our algorithms is the following: If the maximum degree in the graph is small, \local algorithms can be simulated efficiently in the low-memory \mpc model. 
Our method thus basically boils down to reducing the maximum degree of the input graph, as described in \Cref{thmDegRed}, and 
correspondingly consists of two parts: a \emph{degree reduction} part followed by a \emph{\local simulation} part.
In the degree reduction part, which constitutes the key ingredient of our algorithm, we want to find a partial solution (that is, either a matching or an independent set) so that the \emph{remainder graph}---i.e., the graph after the removal of this partial solution (that is, after removing all matched nodes or after removing the independent set nodes along with all their neighbors)---has smaller degree. 

\begin{lemma}\label{degred}
There are $O\left(\log \log n \cdot \log \log \Delta\right)$-round low-memory \mpc algorithms that compute a matching and an independent set in a graph with arboricity $\arb=o(\poly (n))$ so that the remainder graph w.h.p.\ has maximum degree $O\left(\left(\max\{\lambda, \log n\}\right)^{20}\right)$.
\end{lemma}
Note that this directly implies \Cref{thmDegRed}. Next, we show how \Cref{thmMIS,thmMM} follow from \Cref{degred} as well as from an efficient simulation of \local algorithms due to \cite{GU18}.
\begin{proof}[Proof of \Cref{thmMIS,thmMM}]
If $\lambda$, and hence $\Delta$, is at least polynomial in $n$, we directly apply the algorithm by \cite{GU18}, which runs in $O(\sqrt{\log \Delta}+\sqrt{\log \log n})=O(\sqrt{\log n})$ rounds. 
Otherwise, we first apply the algorithm of \Cref{degred} to obtain a partial solution that reduces the degree in the remainder graph to $\Delta'=O(\arb^{20})$ if $\arb\geq\log n$, or to $\Delta'=O(\log^{20} n)$ if $\arb\leq\log n$. It runs in $O(\log \log n \cdot \log \log \Delta)$ rounds. We then apply the algorithm by \cite{GU18} on the remainder graph. This takes $O(\sqrt{\log \Delta'}+\sqrt{\log\log n})=O(\sqrt{\log \lambda}+\sqrt{\log\log n})$ rounds. 
\end{proof} 

\begin{remark}\label{remarkKnowledgeArb}
While our algorithms, at first sight, seem to need to know $\lambda$, we can employ the standard technique \cite{knowledge} of running the algorithm with doubly-exponentially increasing estimates for $\lambda$.
\end{remark}

Our degree reduction algorithm in \Cref{degred} consists of several phases, each reducing the maximum degree by a polynomial factor, as long as the degree is still large enough.
\begin{lemma}\label{poldegred}
There are $O\left( \log \log n\right)$-round low-memory \mpc algorithms that compute a matching and an independent set, respectively, in a graph with arboricity $\arb=o(\poly(n))$ and maximum degree $\Delta=\Omega\left(\left(\max\{\lambda, \log n\}\right)^{20}\right)$ so that the remainder graph w.h.p.\ has maximum degree $O(\Delta^{0.4})$.\end{lemma}
We first show that indeed iterated applications of this polynomial degree reduction lead to the desired degree reduction in \Cref{degred}. 
\begin{proof}[Proof of \Cref{degred}]
We iteratively apply the polynomial degree reduction from \Cref{poldegred}, observing that as long as the maximum degree is still in $\Omega(\lambda^{20})$ and $\Omega(\log^{20}n)$, we reduce the maximum degree by a polynomial factor from $\Delta$ to $O(\Delta^{0.4})$ in each phase, resulting in at most $O(\log \log \Delta)$ phases. 
\end{proof}
It remains to show that such a polynomial degree reduction, as claimed in \Cref{poldegred}, indeed is possible. This is done in two parts. First, in \Cref{sec:seqDegRed}, we provide a centralized algorithm for a polynomial degree reduction, and then, in \Cref{sec:exp}, we show how to implement this centralized algorithm efficiently in the low-memory \mpc model.

\section{A Centralized Degree Reduction Algorithm}\label{sec:seqDegRed}
In this section, we present a centralized algorithm for the polynomial degree reduction as stated in \Cref{poldegred}. For details on how this algorithm can be implemented in the low-memory \mpc model, we refer to \Cref{sec:exp}. 
In \Cref{algoDescription}, we give a formal description of the (centralized) algorithm. Then, in \Cref{correctness}, we prove that this algorithm indeed leads to a polynomial degree reduction. 

\subsection{Algorithm Description}\label{algoDescription}

In the following, we set $\param=\Delta^{1/10}$, and observe that $\param=\Omega(\arb^2)$ as well as $\param=\Omega(\log^{2} n)$, due to the assumptions on $\Delta$ in the lemma statement.
We present an algorithm that reduces the maximum degree to $O(\param^4)$.
This algorithm consists of three phases: a \emph{partition} phase, in which the vertices are partitioned into layers so that every node has at most $d$ neighbors in higher-index layers, 
 a \emph{mark-and-propose} phase in which a random set of candidates is proposed independently in every layer, and a \emph{selection} phase in which a valid subset of the candidate set is selected as partial solution by resolving potential conflicts across layers. 

\paragraph{Partition Phase}
We compute an $H$-partition, that is, a partition of the vertices into layers so that every vertex has at most $d$ neighbors in layers with higher (or equal) index \cite{nash1961edge,nash1964decomposition,barenboim2010sublogarithmic}. 

\begin{definition}[$H$-Partition]
An \emph{$H$-partition} with out-degree $\param$, defined for any $\param>2\arb$, is a partition of the vertices into $\ell$ \emph{layers} $L_1, \dotsc, L_{\ell}$ with the property that a vertex $v\in L_i$ has at most $\param$ neighbors in $\bigcup_{j= i}^{\ell} L_j$. We call $i$ the \emph{layer index} of $v$ if $v \in L_i$. For neighbors $u\in L_i$ and $v\in L_j$ for $i<j$, we call $v$ a \emph{parent} of $u$ and $u$ a \emph{child} of $v$. If we think of the edges as being directed from children to parents, this gives rise to a partial orientation of the edges, with no orientation of the edges connecting vertices in the same layer.
\end{definition}
Note that, for $\param> 2\arb$, such a partition can be computed easily by the following sequential greedy algorithm, also known as \emph{peeling} algorithm: Iteratively, for $i\geq 1$, put all remaining nodes with remaining degree at most $d$ into layer $i$, and remove them from the graph. 

\paragraph{Mark-and-Propose Phase} We first mark a random set of candidates (either edges or vertices) and then propose a subset of these marked candidates for the partial solution as follows.

In the case of maximal matching, every node first marks an outgoing edge chosen uniformly at random and then proposes one of its incoming marked edges, if any, uniformly at random.

In the case of maximal independent set, every node marks itself independently with probability $p=d^{-2}$. Then, if a node is marked and none of its neighbors in the same layer is marked, this node is proposed. 
Note that whether a marked node gets proposed only depends on nodes in the same layer, thus on neighbors with respect to unoriented edges. 

\paragraph{Selection Phase} The set of proposed candidates might not be a valid solution, meaning that it might have some conflicts (i.e., two incident edges or two neighboring vertices). In the selection phase, possible conflicts are resolved (deterministically) by picking an appropriate subset of the proposed candidates, as follows. Iteratively, for $i=\ell, \dotsc, 1$, all (remaining) proposed candidates in layer $i$ are added to the partial solution and then removed from the graph. 

In the case of maximal matching, we add all (remaining) proposed edges directed to a node in layer $i$ to the matching and remove both their endpoints from the graph. 

In the case of maximal independent set, we add all (remaining) proposed nodes in layer $i$ to the independent set and remove them along with their neighbors from the graph. 

\subsection{Proof of Correctness}\label{correctness}

\begin{figure}[!htb]
\caption{Illustration of the mark-and-propose and the selection phase for matching in (a) and independent set in (b). Blue indicates marked but not proposed, green stands for (marked and) proposed but not selected, and red means (marked and proposed and) selected. Note that we omitted all (but a few) irrelevant edges from the figure; the partition into layers thus might not correspond to a valid $H$-partition.}
\centering
\begin{subfigure}{1\textwidth}

		\centering

\includegraphics[scale=1.1]{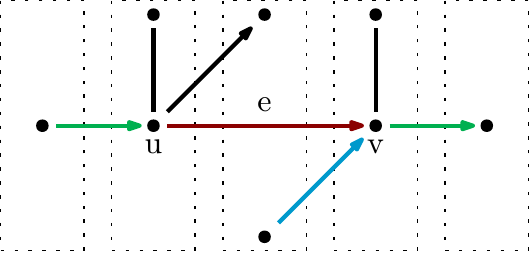}
		\caption{An (oriented) edge $e=(u,v)$ that is selected to be added to the matching cannot have an incident edge that is also selected: an unoriented incident edge cannot be marked as only oriented edges are marked; an oriented edge with the same starting point $u$ cannot be marked as $u$ marks only one outgoing edge; an oriented edge with the same endpoint $v$ cannot be proposed as $v$ proposes only one incoming edge; all other oriented edges $f$ are either processed before (in the case of an outgoing edge from $v$) or after (in the case of an incoming edge to $u$) edge $e$ in the selection phase. In the former case, the selection of $f$ would lead to the removal of $e$ before $e$ is processed; $e$ thus would not be selected. In the latter case, the edge $f$ is removed immediately after $e$ is selected (and thus before $f$ is processed), and thus cannot be selected. 
		}\label{fig:matching}

	\end{subfigure}
	\begin{subfigure}{1\textwidth}

		\centering

	\includegraphics[scale=1.1]{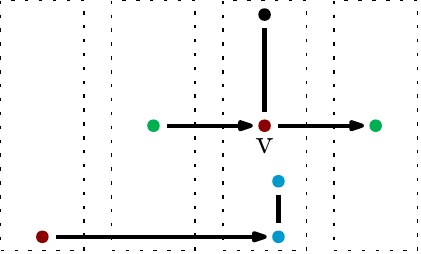}

		\caption{If two neighboring nodes are marked, none of them will be proposed, and consequently, none of them will be selected. A node $v$ that is selected to be added to the independent set cannot have a neighbor that is also selected: a neighbor in the same layer cannot be marked as otherwise $v$ would not be proposed; a neighbor in a lower-index layer is removed from the graph when $v$ joins the independent set, and hence before it potentially could be selected; a selected neighbor in a higher-index layer would lead to $v$'s immediate removal from the graph; when $v$'s layer is processed, $v$ is not part of the graph anymore, and thus could not be selected. 
 }\label{fig:MIS}

	\end{subfigure}

\label{fig:mps}
\end{figure}

It is easy to see that the selected solution is a valid partial solution, that is, that there are no conflicts. It remains to be shown that the degree indeed drops to $O(d^4)$. As the out-degree is bounded by $d$, it is enough to show the following. 

\begin{lemma}\label{lemmaHeavyDisappear}
Every vertex with in-degree at least $d^{4}$ gets removed or all but $d^4$ of its incoming edges get removed, with high probability.
\end{lemma}

\begin{proof}[Proof of \Cref{lemmaHeavyDisappear} for matching]
Let $v$ be a node with degree at least $d^4$. First, observe that if at least one incoming edge of $v$ is proposed, then an edge incident to $v$ (not necessarily incoming) will be selected to be added to the matching. This is because the only reason why $v$ would not select a proposed incoming edge is that $v$ has already been removed from the graph, and this happens only if its proposed edge has been selected to be added to the matching in a previous step. It thus remains to show that every vertex $v$ with in-degree at least $d^4$ with high probability will have at least one incoming edge that has been proposed by the respective child. As every incoming edge of $v$ is proposed independently with probability at least $1/\param$, the probability of $v$ not having a proposed incoming edge is at most $\left(1-1/\param\right)^{\param^{4}}\leq e^{-\param^{3}}=e^{-\Omega(\log^{6}n)}=o\left(\frac{1}{\poly(n)}\right)$. A union bound over all vertices with degree at least $\param^4$ concludes the proof. 
\end{proof}

\begin{proof}[Proof of \Cref{lemmaHeavyDisappear} for independent set]
Let $v$ be a vertex in layer $i$ that is still in the graph and has at least $d^4$ children after all layers with index $\geq i$ have been processed. We show that then at least one of these children will be selected to join the independent set with high probability. Note that this then concludes the proof, as in all the cases either $v$ will be removed from the graph or will not have high in-degree anymore. 
Moreover, observe that such a child $u$ of $v$ (that is still there after having processed layers $\geq i$) will be selected to join the independent set iff it is proposed. This is because if it did not join even though it is proposed, then a parent of $u$ would had been selected to join the independent set, in which case $u$ would not have been part of the graph anymore, at latest after $i$ has been processed, thus would not count towards $v$'s high degree at that point.
 
Every such child $u$ of $v$ is marked independently with probability $p=d^{-2}$. The probability of $u$ being proposed and hence joining the independent set is at least $p(1-p)^{d}$, as it has at most $d$ neighbors in its layer, and it is proposed iff it is marked and none of its neighbors in the same layer is marked. Vertex $v$ thus in expectation has at least $\mu:=d^{4} p(1-p)^{\param}\geq d^{2}e^{-2/d}=\Omega(\param^{2})$ children that join the independent set. 

Since whether a node $u$  proposes and hence joins the independent set depends on at most $\param$ other nodes (namely $u$'s neighbors in the same layer), it thus follows from a variant of the Chernoff bounds for bounded dependence, e.g., from Theorem 2.1 in \cite{pemmaraju2001equitable}, that the probability of $v$ having, say, $0.5\mu$ neighbors that join the independent set is at most $e^{-\Omega(\mu/\param)}=e^{-\Omega(\param)}\leq e^{-\Omega\left(\log^{2}n\right)}= o\left(\frac{1}{\poly(n)}\right)$. A union bound over all vertices $v$ with degree at least $d^4$ concludes the proof.\end{proof}

\section{Implementation of the Degree Reduction Algorithm in \mpc}\label{sec:exp}
In this section, we show how to simulate the centralized degree reduction algorithm from \cref{sec:seqDegRed} in the low-memory \mpc model.
In \Cref{sec:construction}, we show how to implement the partition phase efficiently in the low-memory \mpc model. Then, in \cref{sec:simul}, we describe how to perform the simulation of the mark-and-propose as well as the selection phase.
Together with the correctness proof established in \Cref{correctness}, this will conclude the proof of \Cref{poldegred}. 

The main idea behind the simulation is to use the well-known \emph{graph exponentiation} technique ~\cite{Lenzen2010brief}, which can be summarized as follows: 
Suppose that every node knows its $2^{i - 1}$-hop neighborhood in iteration $i - 1$.
Then, in iteration $i$, each node can inform the nodes in its $2^{i - 1}$-hop neighborhood of the topology of its $2^{i - 1}$-hop neighborhood.
Hence, every node can learn of its $2^{i}$-hop neighborhood in iteration $i$, allowing it to simulate any $2^{i}$-round \local algorithm in $0$ rounds. Using this exponentiation technique, in principle, every $t$-round \local algorithm can be simulated in $O(\log t)$ \mpc rounds. We have to be careful about the memory restrictions though.
In iteration $i$ of the exponentiation process, we need to store a copy of the $2^{i}$-hop neighborhood of every node.
If the neighborhood contains more than $n^{\delta}$ edges, we violate the local memory constraint.
Similarly, the total memory of $\widetilde \Theta(m)$ might not be enough to store each of these copies, even if every single copy fits to a machine.
In order to deal with these issues, the key observation is that a large fraction of the nodes in the graph are contained in the first layers of the $H$-partition.
In particular, we show that if we focus on the graph remaining after $\ell/2$ iterations of peeling, we can perform roughly $\log \ell$ exponentiation steps without violating the memory constraints.
Hence, we can perform $2^i$ steps of the degree reduction process in roughly $i$ communication rounds.

	In case the maximum degree $\Delta$ is larger than the local memory $S=O(n^{\delta})$, one needs to pay attention to how the graph is distributed.
	One explicit way is to split high-degree nodes into many copies and distribute the copies among many machines.
	For the communication between the copies, one can imagine a (virtual) balanced tree of depth $1 / \delta$ rooted at one of the copies.
	Through this tree, the copies can exchange information in $O(1/\delta)$ communication rounds.
	For the sake of simplicity, our write-up assumes that $\Delta \ll n^{\delta}$. 

\subsection{Partition Phase}\label{sec:construction}

We first prove some properties of the $H$-partition constructed by the greedy peeling algorithm that will be useful for an efficient implementation in the low-memory \mpc model. 

\begin{lemma}\label{HDecompProp}
The $H$-partition with out-degree $\param$, constructed by the greedy peeling algorithm, satisfies the following properties.
\begin{enumerate}[(i)]
\item For all $0\leq i \leq \ell$, the number $\left|\bigcup_{j=i}^{\ell}L_j\right|$ of nodes in layers with index $\geq i$ is at most $n \left(\frac{2\arb}{\param}\right)^{i-1}$. In other words, if we remove all nodes in layer $i$ from the set of nodes in layers $\geq i$, then the number of nodes drops by a factor of $\frac{2\arb}{\param}$, i.e., $\left|\bigcup_{j=i+1}^{\ell}L_j\right|\leq  \frac{2\arb}{\param}\left|\bigcup_{j=i}^{\ell}L_j\right|$ for all $0\leq i \leq \ell$. 
\item There are at most $\ell=O\left(\log_{\frac{\param}{\arb}}n\right)$ layers.
\end{enumerate}
\end{lemma}
\begin{proof}
We prove (i) by induction, thus assume that there are $n_{i}\leq n \left(\frac{2\arb}{\param}\right)^{i-1}$ nodes in the graph $H_i$ induced by vertices in layers $\geq i$. Towards a contradiction, suppose that there are $n_{i+1}>n \left(\frac{2\arb}{\param}\right)^{i}$ nodes in layers $\geq i+1$. By construction, all these nodes must have had degree larger than $\param$ in $H_i$, as otherwise they would have been added to layer $i$. This results in an average degree of more than $\frac{n_{i+1} \param}{n_i} =2\arb$ in $H_i$, which contradicts the well-known upper bound of $2\arb$ on the average degree in a graph that has arboricty at most $\arb$. Note that (ii) is a direct consequence of (i). 
\end{proof}

In the following, we describe how to compute the $H$-partition with parameter $\param = \Delta^{1/10}$ in the low-memory \mpc model\footnote{Note that it is easy to learn the maximum degree, and hence $\param$, in $O(1)$ rounds of communications.}.
Throughout this section, we assume that $\Delta \geq (2\arb)^{20}$, i.e., that $\param \geq (2\arb)^2$.
Observe that if $\Delta^2 > n^\delta$, then the $H$-partition with parameter $\param = \Delta^{1/10}$ consists of $O(\log_{\param/\arb}n) = O(\log_{\Delta}n) = O(1/\delta)$ layers, in which case the arguments in this section imply that going through the layers one by one will easily yield at least as good runtimes as for the more difficult case of $\Delta^2 \leq n^\delta$.
Hence, throughout this section, we assume that $\Delta^2 \leq n^\delta$. 

The goal of the algorithm for computing the $H$-partition is that each node (or, more formally, the machine storing the node) knows in which layer of the $H$-partition it is contained.
The algorithm proceeds in iterations, where each iteration consists of two parts: first, the output, i.e., the layer index, is determined for a large fraction of the nodes, and second, these nodes are removed for the remainder of the computation. The latter ensures that the remaining small fraction of nodes can use essentially all of the total available memory in the next iteration, resulting in a larger memory budget per node.
However, there is a caveat: When the memory budget per node exceeds $\Theta(n^\delta)$, i.e., the memory capacity of a single machine, then it is not sufficient anymore to merely argue that the used memory of all nodes together does not exceed the total memory of all machines together\footnote{Furthermore, in the ``shuffle'' step of every MPC round~\cite{karloff2010model}, we assume that the nodes are stored in the machines in a balanced way, i.e., as long as a single node fits onto a single machine and the total memory is not exceeded, the underlying system takes care of load-balancing.}.
We circumvent this issue by starting the above process repeatedly from anew (in the remaining graph) each time the memory requirement per node reaches the memory capacity of a single machine.
As we will see, the number of repetitions, called \emph{phases}, is bounded by $O(1/\delta)$.

\newpage
In the following, we examine the phases and iterations in more detail.

\paragraph{Algorithm Details}
Let $k$ be the largest integer s.t.\ $\Delta^{2^k+1}\leq n^{\delta}$ (which implies that $k \geq 0$).
The algorithm consists of phases and each phase consists of $k+1$ iterations.
Next, we describe our implementation of the graph exponentiation in more detail.
\begin{itemize}
	\item  In each iteration $i = 0, 1, \ldots, k$, we do the following.
	\begin{itemize}
		\item Let $G_i = G_i^{(0)}$ be the graph at the beginning of iteration $i$. Each node connects its current $1$-hop neighborhood to a clique by adding virtual edges to $G_i$; if $i = 0$, omit this step. 
Perform $20$ repetitions of the following process if $i \geq 1$, and $60$ repetitions if $i=0$:
		\begin{itemize}
			\item In repetition $0 \leq j \leq 19$ (resp.\ $0 \leq j \leq 59$), each node computes its layer index in the $H$-partition of $G_i^{(j)}$ (with parameter $\param$) or determines that its layer index is strictly larger than $2^i$, upon which all nodes in layer at most $2^i$ (and all its incident edges) are removed from the graph, resulting in a graph $G_i^{(j+1)}$.
		\end{itemize}
		Set $G_{i+1} = G_i^{(20)}$ (resp.\ $G_{i+1} = G_i^{(60)}$ if $i=0$). 
	\end{itemize}
	At the end of the phase remove all added edges.
\end{itemize}
The algorithm terminates when each node knows its layer.

Note that each time a node is removed from the graph, the whole layer that contains this node is removed, and each time such a layer is removed, all layers with smaller index are removed at the same time or before.
By the definition of the $H$-partition, if we remove the $\ell$ layers with smallest index from a graph, then there is a 1-to-1 correspondence between the layers of the resulting graph and the layers with layer index at least $\ell+1$ of the original graph.
More specifically, layer $\ell'$ of the resulting graph contains exactly the same nodes as layer $\ell+\ell'$ of the original graph.
Hence, if a node knows its layer index in some $G_i^{(j)}$, it can easily compute its layer index in our original input graph $G$, by keeping track of the number of deleted layers, which is uniquely defined by $i$, $j$ and the number of the phase.
We implicitly assume that each node performs this computation upon determining its layer index in some $G_i^{(j)}$ and in the following only consider how to determine the layer index in the current graph.

\paragraph{Implementation in the \mpc Model}
Let us take a look at one iteration. 

Connecting the $1$-hop neighborhoods to cliques is done by adding the edges that are missing.
Edges that are added by multiple nodes are only added once (since the edge in question is stored by the machines that contain an endpoint of the edge, this is straightforward to realize).
Note that during a phase, the $1$-hop neighborhoods of the nodes grow in each iteration (if not too many close-by nodes are removed from the graph); more specifically, after $i$ iterations of connecting $1$-hop neighborhoods to cliques, the new $1$-hop neighborhood of a node contains exactly the nodes that were contained in its $2^i$-hop neighborhood at the beginning of the phase (and were not removed so far).

In iteration $i$, the layer of a node is computed as follows: First each node locally gathers the topology of its $2^i$-hop neighborhood (without any added edges).\footnote{Note that it is easy to keep track of which edges are original and which are added, incurring only a small constant memory overhead; later we will also argue why storing the added edges does not violate our memory constraints.}
Since this step is performed after connecting the $2^{i-1}$-hop neighborhood of each node to a clique (by repeatedly connecting $1$-hop neighborhoods to cliques), i.e., after connecting each node to any other node in its $2^i$-hop neighborhood, only $1$ round of communication is required for gathering the topology.
Moreover, since a node that knows the topology of its $2^i$-hop neighborhood can simulate any $(2^i - 1)$-round distributed process locally, it follows from the definition of the $H$-partition, that knowledge of the topology of the $2^i$-hop neighborhood is sufficient for a node to determine whether its layer index is at most $2^i$ and, if this is the case, in exactly which layer it is contained.
Thus, the only tasks remaining are to bound the runtime of our algorithm and to show that the memory restrictions of our model are not violated by the algorithm.

\paragraph{Runtime}
It is easy to see that every iteration takes $O(1)$ rounds. 
Thus, in order to bound the runtime of our algorithm, it is sufficient to bound the number of iterations by $O(1/\delta \cdot\log \log n)$.
	By \cref{HDecompProp} (ii) the number of layers in the $H$-partition of our original input graph $G$ is $O(\log_{\param/\arb}n)$, which is $O(\log_{\Delta}n)$ since $\param/(2\arb) \geq \sqrt{\param} = \Delta^{1/20}$.
	Consider an arbitrary phase.
	According to the algorithm description, in iteration $i \geq 1$, all nodes in the $20 \cdot 2^i$ lowest layers are removed from the current graph.
	Hence, ignoring iteration $0$, the number of removed layers doubles in each iteration, and we obtain that the number of layers removed in the $k+1$ iterations of our phase is $\Omega(2^k)$.
	By the definition of $k$, we have $\Delta^{2^{k+1}+1} > n^{\delta}$, which implies $2^k > 1/3 \cdot \delta \cdot \log_{\Delta}n$.

	Combining this inequality with the observations about the total number of layers and the number of layers removed per phase, we see that the algorithm terminates after $O(1/\delta)$ phases.
	Since there are $k+1 = O(\log \log n)$ iterations per phase, the bound on the number of iterations follows. 

\paragraph{Memory Footprint}
As during the course of the algorithm edges are added and nodes collect the topology of certain neighborhoods, we have to show that adding these edges and collecting these neighborhoods does not violate our memory constraints of $O(n^\delta)$ per machine.
As a first step towards this end, the following lemma bounds the number of nodes contained in graph $G_i$.

\begin{lemma} \label{lem:nodedecrease}
	Consider an arbitrary phase. Graph $G_i$ from that phase contains at most $n'/(\Delta^{2^i})$ nodes, for all $i \geq 1$, where $n' = n/\Delta$.
\end{lemma}

\begin{proof}
	By \cref{HDecompProp} (i), removing the nodes in the layer with smallest index from the current graph decreases the number of nodes by a factor of at least $\param/(2\arb) \geq \param^{1/2} = \Delta^{1/20}$.
	We show the lemma statement by induction.
	Since in iteration $0$ the nodes in the $60$ layers with smallest index are removed, we know that $G_1$ contains at most $n/(\Delta^3) = n'/(\Delta^{2^1})$ nodes.
	Now assume that $G_i$ contains at most $n'/(\Delta^{2^i})$ nodes, for an arbitrary $i \geq 1$.
	According to the design of our algorithm, $G_{i+1}$ is obtained from $G_i$ by removing the nodes in the $20 \cdot 2^i$ layers with smallest index.
	Combining this fact with our observation about the decrease in the number of nodes per removed layer, we obtain that $G_{i+1}$ contains at most $n'/(\Delta^{2^i}) \cdot 1/(\Delta^{2^i}) = n'/(\Delta^{2^{i+1}})$ nodes.
\end{proof}
Using \cref{lem:nodedecrease}, we now show that the memory constraints of the low-memory \mpc model are not violated by our algorithm.
Consider an arbitrary phase and an arbitrary iteration $i$ during that phase.
If $i=0$, then no edges are added and each node already knows the topology of its $2^i$-hop neighborhood, so no additional memory is required.
Hence, assume that $i \geq 1$.

Due to \cref{lem:nodedecrease}, the number of nodes considered in iteration $i$ is at most $n'/(\Delta^{2^i})$, where, again, $n' = n/\Delta$.
After the initial step of connecting $1$-hop neighborhoods to cliques in iteration $i$, each remaining node is connected to all nodes that were contained in its $2^i$-hop neighborhood in the original graph $G$ (and were not removed so far).
Hence, each remaining node is connected to at most $O(\Delta^{2^i})$ other nodes, resulting in a memory requirement of $O(\Delta^{2^i})$ per node, or $O(n/\Delta)$ in total.
Similarly, when collecting the topology of its $2^i$-hop neighborhood, each node has to store $O(\Delta^{2^i} \cdot \Delta)$ edges, which requires at most $O(\Delta^{2^i} \cdot \Delta)$ memory, resulting in a total memory requirement of $O(n)$.
Hence, the described algorithm does not exceed the total memory available in the low-memory \mpc model.
Moreover, due to the choice of $k$, the memory requirement of each single node does not exceed the memory capacity of a single machine.

\newpage
\subsection{Simulation of the Mark-and-Propose and Selection Phase}\label{sec:simul}
For the simulation of the mark-and-propose and selection phase, we rely heavily on the approach of \cref{sec:construction}.
Recall that nodes were removed in \emph{chunks} consisting of several consecutive layers and that before a node $v$ was removed, $v$ was directly connected to all nodes contained in a large neighborhood around $v$ by adding the respective edges.
For the simulation, we go through these chunks in the reverse order in which they were removed.
Note that in which chunk a node is contained is uniquely determined by the layer index of the node.
As each node computes its layer index during the construction of the $H$-partition, each node can easily determine in which part of the simulation it will actively participate.

However, there is a problem we need to address:
For communication, we would like the edges that we added during the construction of the $H$-partition to be available also for the simulation.
Unfortunately, during the course of the construction, we removed added edges again to free memory for the adding of other edges.
Fortunately, there is a conceptually simple way to circumvent this problem: in the construction of the $H$-partition, add a pre-processing step in the beginning, in which we remove the lowest $c \cdot \log(1/\delta \cdot \log \log n)$ layers (where $c$ is a sufficiently large constant) one by one in $\log(1/\delta \cdot \log \log n)$ rounds, which increases the available memory (compared to the number of (remaining) nodes) by a factor of $\Omega(1/\delta \cdot \log \log n)$, by \cref{HDecompProp}.
Since the algorithm for constructing the $H$-partition consist of $O(1/\delta \cdot \log \log n)$ iterations, this implies that we can store all edges that we add during the further course of the construction simultaneously without violating the memory restriction, by an argument similar to the analogous statement for the old construction of the $H$-partition.
Similarly, also the number of added edges incident to one particular node does not exceed the memory capacity of a single machine.
In the following, we assume that this pre-processing step took place and all edges added during the construction of the $H$-partition are also available for the simulation.

\paragraph{Matching Algorithm}
As mentioned above, we process the chunks one by one, in decreasing order w.r.t.\ the indices of the contained layers.
After processing a chunk, we want each node contained in the chunk to know the output of all incident edges according to the centralized matching algorithm. In the following, we describe how to process a chunk, after some preliminary ``global" steps.

The mark-and-propose phase of the algorithm is straightworward to implement in the low-memory \mpc model: each node (in each chunk at the same time) performs the marking of an outgoing edge as specified in the algorithm description.\footnote{Note that, formally, the algorithm for construction the $H$-partition only returns the layer index for each node; however, from this information each node can easily determine which edges are outgoing, unoriented, or incoming according to the partial orientation induced by the $H$-partition.}
The proposing is performed for all nodes before going through the chunks sequentially: each node proposes one of its marked incoming edges (it there is at least one) uniformly at random.
Note that proposes an edge does not necessarily indicate that this edge will be added to the matching; more specifically, an edge proposed by some node $v$ will be added to the matching iff the edge that $v$ marked is not selected to be added to the matching.\footnote{In other words, only proposed edges can go into the matching and whether such an edge indeed goes into the matching can be determined by going through the layers in decreasing order and only adding a proposed edge if there is no conflict.}

After this mark-and-propose phase, the processing of the chunks begins.
Consider an arbitrary chunk.
Let $i$ be the iteration (in some phase) in which this chunk was removed in the construction of the $H$-partition, i.e., the chunk consists of $2^i$ layers.  
Each node in the chunk collects the topology of its $2^i$-hop neighborhood in the chunk including the information contained therein about proposed edges.
Due to the edges added during the construction of the $H$-partition, this can be achieved in a constant number of rounds, and by an analogous argument to the one at the end of \cref{sec:construction}, collecting the indicated information does not violate the memory restrictions of our model.
\cref{lem:matchingimp} shows that the information contained in the $2^i$-hop neighborhood of a node is sufficient for the node to determine the output for each incident edge in the centralized matching algorithm.

\begin{lemma} \label{lem:matchingimp}
	The information about which edges are proposed in the $2^i$-hop neighborhood of a node $v$ uniquely determines the output of all edges incident to $v$ according to the centralized matching algorithm.
\end{lemma}

\begin{proof}
	From the design of the centralized matching algorithm, it follows that an edge is part of the matching iff 1) the edge is proposed and 2) either the higher-layer endpoint of the edge has no outgoing edges or the outgoing edge marked by the higher-layer endpoint is not part of the matching.
	Hence, in order to check whether an incident edge is in the matching, node $v$ only has to consider the unique directed chain of proposed edges (in the chunk) starting in $v$.
	Clearly, the information which of the edges in this chain are proposed uniquely defines the output of the first edge in the chain, from which $v$ can infer the output of all other incident edges.
	Since the number of edges in the chain is bounded by $2^i - 1$ as the chain is directed, the lemma statement follows.
\end{proof}
It thus follows from the bound on the number of iterations that the simulation of the selection phase for the matching algorithm can be performed in $O(1/\delta \cdot\log \log n)$ rounds of communication.

\paragraph{Independent Set Algorithm}
The simulation of the independent set algorithm proceeds analogously to the case of the matching algorithm.
First, each node performs the marking and proposing in a distributed fashion in a constant number of rounds.
Then, the chunks are processed one by one, as above, where during the processing of a chunk removed in iteration $i$, each node contained in the chunk collects its $2^i$-hop neighborhood, including the information about which nodes are proposed, and then computes its own output locally.
By analogous arguments to the ones presented in the case of the matching algorithm, the algorithm adheres to the memory constraints of our model and the total number of communication rounds is $O(1/\delta \cdot\log \log n)$.
The only part of the argumentation where a bit of care is required is the analogue of \cref{lem:matchingimp}:
In the case of the independent set algorithm the output of a node $v$ may depend on \emph{each} of its parents since each of those could be part of the independent set, which would prevent $v$ from joining the independent set.
However, \emph{all} nodes in the chunk that can be reached from $v$ via a directed chain of edges are contained in $v$'s $2^i$-hop neighborhood; therefore, collecting the own $2^i$-hop neighborhood is sufficient for determining one's output.
Note that at the end of processing a chunk, if we follow the above implementation, we have to spend an extra round for removing the neighbors of all selected independent set nodes since these may be contained in another chunk.

\bibliographystyle{alpha}
\bibliography{ref}

\end{document}